\documentclass[runningheads]{llncs}

\usepackage[fleqn]{amsmath}

\usepackage{amsthm}

\usepackage{amssymb}
\usepackage{amsfonts}
\usepackage{algorithmic}
\usepackage{graphicx}
\usepackage{textcomp}
\usepackage[numbers]{natbib}
\usepackage[pdfpagelabels]{hyperref} 
\usepackage{url}
\usepackage{mathpartir}
\usepackage{mathrsfs} 
\usepackage{semantic} 
\usepackage{color}
\usepackage{mathtools}  
\usepackage{thmtools}
\usepackage{xspace}
\usepackage{xcolor}
   
\usepackage{anyfontsize}
\usepackage{amsbsy}       
\usepackage{stmaryrd}
\usepackage{braket}
\usepackage{marvosym} 
\usepackage{wasysym}    
\usepackage{balance}    
\usepackage{proof}
\usepackage{yhmath}
\usepackage{rotating}
\usepackage{tikz-cd}
\usepackage{changepage}
\usepackage{centernot}
\usepackage{comment}
\usepackage{listings}
\usepackage{mdframed}
\usepackage{centernot}
\usepackage{xparse}
\usepackage{upgreek}

\usepackage{multicol}

\newenvironment{block}[1][t]
  {\begin{array}[#1]{@{}l@{}}}
  {\end{array}}

\definecolor{lightgray}{gray}{0.90}
\newcommand{\Gbox}[1]{\colorbox{lightgray}{$#1$}}

\mathlig{[]}{\square}
\mathlig{|-s}{\vdash^{\!\forall}}
\mathlig{|-}{\vdash}
\mathlig{|}{\mid} 
\mathlig{<=>}{\Longleftrightarrow} 

\reservestyle{\oblang}{\mathsf}
\oblang{if[if\;],then[\;then\;],else[\;else\;],let[let\;],in[\;in\;]}

\newcommand{\Int}{\mathsf{Int}}
\newcommand{\String}{\mathsf{String}}

\newcommand{\Bool}{\mathsf{Bool}}

\newcommand{\StringLen}{\mathsf{StringLen}}

\newcommand{\union}{\cup}

\newcommand{\ie}{\emph{i.e.}\xspace}

\newcommand{\eg}{\emph{e.g.}\xspace}

\mathlig{++}{\mathbin{++}}

\newcommand{\ltop}{\textsf{H}}
\newcommand{\lbot}{\textsf{L}}

\definecolor{lightred}{RGB}{255,100,100}

\newcommand{\icode}[1]{$\mathsf{#1}$}

\newcommand\defas[0]{\stackrel{\triangle}{=}}

\newcommand{\stypetop}[1]{{#1}_{\ltop}}
\newcommand{\stypebot}[1]{{#1}_{\lbot}}

\newcommand{\stype}[2]{{#1\triangleleft#2}}

\newcommand\ssubst[3]{#1\left[#2 / #3\right]}

\newcommand{\reduce}{\longmapsto}

\newcommand{\stypeof}[3]{#1 \vdash_{1} #2 : #3}

\newcommand{\mtype}[2]{#1 \rightarrow #2}

\newcommand{\DeltaX}{\Delta}

\newcommand\wf[2]{#1 \models #2}

\newcommand{\primt}{P}
\newcommand{\primb}{{\mathsf{p}}}

\newcommand{\unionrel}[1]{{\text{Rel}\left[#1\right]}}
\newcommand{\downreln}[1]{\left\lfloor R\right\rfloor_{n}}

\newcommand{\rhosem}[1]{\rho_\mathsf{R}(#1)}

\newcommand{\gsetvx}[2]{{\mathcal{V}\llbracket#1\rrbracket}{#2}}

\newcommand{\gsetc}[1]{\gsetcx{#1}{\rho}}
\newcommand{\gsetcx}[2]{{\mathcal{C}\llbracket#1\rrbracket}{#2}}

\newcommand{\gsetg}[1]{\gsetgx{#1}{\rho}}
\newcommand{\gsetgx}[2]{\mathcal{G}\llbracket#1\rrbracket{#2}}

\newcommand{\gsetd}[1]{\mathcal{D}\llbracket#1\rrbracket}

\newcommand{\xsrho}[3]{#1\left[#2 \mapsto #3 \right]}

\makeatletter
\newcommand*\bigcdot{\mathpalette\bigcdot@{.7}}
\newcommand*\bigcdot@[2]{\mathbin{\vcenter{\hbox{\scalebox{#2}{$\m@th#1\bullet$}}}}}
\makeatother

\newcommand\eType[2]{\exists #1.#2}

\newcommand{\pack}[3]{\mathsf{pack}(#1,#2)~\mathsf{as}~#3}
\newcommand{\packSimple}[2]{\mathsf{pack}(#1,#2)}
\newcommand{\unpack}[4]{\mathsf{open}(#1,#2) = #3~\mathsf{in}~#4}

\newcommand{\twoTypes}[2]{{#1}@{#2}}

\newcommand{\AccountStore}{\mathsf{AccountStore}}

\newcommand{\SalaryPolicy}{\mathsf{SalaryPolicy}}

\newcommand{\optiont}[1]{\mathsf{Option}\left[#1\right]}

\newcommand{\XRel}{R}
\newcommand{\erni}[4]{\mathsf{ERNI}(#1,#2,#3,#4)}

\newcommand\elsec[0]{$\mathsf{\lambda}^{\exists}_{\mathsf{SEC}}$\xspace}
\newcommand\pairt[2]{#1\times#2}
\newcommand\sumt[2]{#1+#2}

\newcommand\paire[2]{\left\langle #1,#2\right\rangle }
\newcommand\pFst[1]{\mathsf{fst}~#1}
\newcommand\pSnd[1]{\mathsf{snd}~#1}

\newcommand\inl[1]{\mathsf{inl}~#1}
\newcommand\inr[1]{\mathsf{inr}~#1}
\newcommand\sCase[5]{\mathsf{case}~#1~\mathsf{of}~\mathsf{inl}~#2.#3|\mathsf{inr}~#4.#5}

\newcommand\lambdae[3]{\lambda#1:#2.~#3}
\newcommand\appe[2]{#1~#2}

\newcommand\binOp[2]{#1\oplus#2}
\newcommand\binOpSym{\oplus}
\newcommand\unitval{\left\langle \right\rangle}
\newcommand\unittype{\mathsf{1}}

\newcommand\repType[1]{\mathsf{sftype}(#1)}

\newcommand\precise[2]{#1 \sqsubseteq #2}

\newcommand\stamp[2]{\lceil#1\rceil_{#2}}

\newcommand\leftFacet{safety~}
\newcommand\rightFacet{declassification~}

\newcommand\publicWord{public~}
\newcommand\secretWord{secret~}

\newcommand{\eatomtwo}[2]{{\text{Atom}\left[#1,#2\right]}} 
\newcommand{\eatomunion}[1]{{\text{Atom}_{\rho}\left[#1\right]}} 
\newcommand{\erel}[2]{{\text{Rel}\left[#1,#2\right]}}

\newcommand{\rhofst}{\rho_{1}}
\newcommand{\rhosnd}{\rho_{2}}

\newcommand{\extenv}[3]{{#1\left[#2 \mapsto #3\right]}}

\newcommand{\esetv}[1]{\esetvx{#1}{\rho}}
\newcommand{\esetvx}[2]{\mathcal{V}\llbracket#1\rrbracket{#2}}

\newcommand{\esetc}[1]{\esetcx{#1}{\rho}}
\newcommand{\esetcx}[2]{{\mathcal{C}\llbracket#1\rrbracket}{#2}}

\newcommand{\esetg}[1]{\esetgx{#1}{\rho}}
\newcommand{\esetgx}[2]{\mathcal{G}\llbracket#1\rrbracket{#2}}

\newcommand{\esetd}[1]{\mathcal{D}\llbracket#1\rrbracket}

\usepackage{inconsolata} 

\definecolor{darkgreen}{RGB}{0,128,0}
\lstdefinelanguage{scala}{
    morekeywords={let,abstract,case,catch,class,def,%
      do,else,extends,false,final,finally,%
      for,if,implicit,import,match,mixin,%
      new,null,object,override,package,%
      private,protected,requires,return,sealed,%
      super,this,throw,trait,true,try,%
      type,val,var,while,with,yield, app, has,
			top,bottom,declassify,Obj},
    sensitive=true,
		keywordstyle={\color{blue}},
    morecomment=[l][\color{darkgreen}]{//},
    morecomment=[n]{/*}{*/},
    morestring=[b]",
    morestring=[b]',
    morestring=[b]""",
    escapeinside={(*}{*)},
    moredelim=**[is][{\btHL}]{`}{`}
  }
\begin{document}
\title{Existential Types for Relaxed Noninterference}
\author{Raimil Cruz\inst{1}\thanks{This work is partially funded by CONICYT FONDECYT Regular Projects 1150017 and 1190058. 
Raimil Cruz is partially funded by CONICYT-PCHA/Doctorado Nacional/2014-63140148} \and
\'Eric Tanter\inst{1}}
\authorrunning{R. Cruz and \'E. Tanter}
\institute{PLEIAD Lab, Computer Science Department (DCC), University of Chile, Santiago, Chile \\
\email{\{racruz,etanter\}@dcc.uchile.cl}}
\maketitle
\begin{abstract}
Information-flow security type systems ensure confidentiality by enforcing noninterference: a program cannot leak private data to public channels. However, in practice, programs need to selectively declassify information about private data. 
Several approaches have provided a notion of {\em relaxed noninterference} supporting selective and expressive declassification while retaining 
a formal security property. The {\em labels-as-functions} 
approach provides relaxed noninterference by means of declassification policies expressed as functions. The {\em labels-as-types} approach expresses declassification policies using type abstraction and {\em faceted types}, a pair of types representing the secret and public facets of values. The original proposal of labels-as-types is formulated in an object-oriented setting where type abstraction is realized by subtyping. 
The object-oriented approach however suffers from limitations due to its receiver-centric paradigm. 

In this work, we consider an alternative approach to labels-as-types, applicable in non-object-oriented languages, which allows us to express advanced declassification policies, such as extrinsic  policies, based on a different form of type abstraction: existential types. 
An existential type exposes abstract types and operations on these; we leverage this abstraction mechanism to express secrets that can be declassified using the provided operations. 
We formalize the approach in a core functional calculus with existential types, define existential relaxed noninterference, and prove that well-typed programs satisfy this form of type-based relaxed noninterference.

\end{abstract}
\section{Introduction}
\label{sec:elsec-introduction}
A sound information-flow security type system ensures confidentiality by means of {\em noninterference}, a property that states that \publicWord 
values (\eg $\String_{\lbot}$) do not depend on \secretWord values (\eg $\String_{\ltop}$). This enables a modular
reasoning principle about well-typed programs. For instance, in a pure language, a function $f: \String_{\ltop} -> \String_{\lbot}$ is necessary a constant function because the (public) result cannot leak information about the (private) argument.

However, noninterference is too stringent and real programs need to explicitly {\em declassify} some information about \secretWord values. A simple mechanism to support explicit declassification is to add a \icode{declassify} operator from \secretWord to \publicWord expressions, as provided for instance in Jif~\cite{myers:jif}. However, 
the arbitrary use of this operator breaks formal guarantees about confidentiality. Providing a declassification mechanism while still enforcing a noninterference-like property is an active topic of research~\cite{cruzAl:ecoop2017,cruzTanter:secdev2019,hicksAl:plas2006,liZdancewic:popl2005,ngoAl:arXiv2019,sabelfeldSands:jcs2009}.

One interesting mechanism is the {\em labels-as-functions} approach of \citet{liZdancewic:popl2005}, which supports 
{\em declassification policies} while ensuring {\em relaxed} noninterference. 
Instead of using security labels such as $\lbot$ and $\ltop$ that are taken from a 
security lattice of symbols, security labels are functions. These functions, called declassification policies, denote the intended computations to declassify values. For instance, the function $\lambda x. \lambda y. x ==y$ denotes the declassification policy: ``the result of the comparison
of the \secretWord value $x$ with the \publicWord value $y$ can be declassified''. The identity 
function denotes \publicWord values, while a constant function denotes \secretWord values. Then, any use of a value that
does not follow its declassification policy yields a \secretWord result. The labels-as-functions approach is very expressive, 
but its main drawback is that label ordering relies on a semantic interpretation of functions and program equivalence, which is hard to realize in practice and rules out recursive declassification policies\footnote{\citet{liZdancewic:popl2005} rule out recursive declassification because otherwise the subtyping relation induced by security labels (sets of functions) would be undecidable.}.

An alternative approach to {\em labels-as-functions} is {\em labels-as-types}, recently proposed by \citet{cruzAl:ecoop2017}. The key idea is to exploit type abstraction to control how much of a value is open to declassification. 
The approach was originally developed in an object-oriented language, where type abstraction is realized by subtyping. A security type $\stype{T}{U}$ is composed of two facets: the safety type $T$ denotes the secret view of the value, and the declassification type $U$ (such that $T<:U$) specifies the public view, \ie~the methods that can be used to declassify a \secretWord value.
For instance, the type $\stype{\String}{\String}$ denotes a \publicWord string value, \ie all the methods of $\String$ are 
available for declassification, while 
the type $\stype{\String}{\top}$ (where $\top$ is the empty interface type) denotes a secret $\String$ value, 
\ie there is no method available
to declassify information about the secret. 
Then, the interesting declassification policies are expressed with a type interface between $\String$  and $\top$; \eg~the type $\stype{\String}{\StringLen}$ exposes the method \icode{length} of $\String$ for declassification.
With this type-based approach, label ordering is simplified to standard subtyping, which is a simple syntactic property, and naturally supports recursive declassification. Also, this type-based approach enforces
a security property called {\em type-based relaxed noninterference}, which accounts for type-based declassification and provides a modular reasoning principle similar to standard noninterference.

We observe that developing type-based relaxed noninterference in an object-oriented setting, exploiting subtyping as the type abstraction mechanism, imposes some restrictions on the declassification policies that can be expressed. In particular, because security types are of the form $\stype{T}{U}$ where the declassification type $U$ is a supertype of the safety type $T$---a necessary constraint to ensure type safety---means that one cannot declassify properties that are {\em extrinsic} to (\ie~computed externally from) the secret value. For instance, because a typical $\String$ type does not feature 
an \icode{encrypt} method, it is not possible to express the declassification policy that ``the encrypted representation of the password is public". 

%
In this paper, we explore an alternative approach to labels-as-types and relaxed noninterference, exploiting another well-known type abstraction mechanism: existential types. An existential type $\exists X.T$ provides an abstract type $X$ and an interface $T$ to operate with values of the abstract type $X$. Then instances of the abstract type $X$ are akin to secrets that can be declassified using the operations described by $T$. For instance, the existential type $\exists X.\lbrack \mathsf{get}: X, \mathsf{length}: X -> \Int \rbrack$ makes it possible to obtain a (secret) value of type $X$ with \icode{get}, that only can be ``declassified'' with the \icode{length} function to obtain a (public) integer.

Because existential types are the essence of abstraction mechanisms like abstract data types and modules~\cite{mitchellPlotkin:toplas1888}, this work shows how the labels-as-types approach can be applied in non-object-oriented languages. The only required extension is the notion of faceted types, which are necessary to capture the natural separation 
between {\em privileged observers} (allowed to observe secret results) and {\em public observers} (\ie the attacker, which can only observe public values)~\footnote{To account for $n > 2$ observation levels, faceted types can be extended to have n facets.}. 
Additionally, the existential approach is more expressive than the object-oriented one in that extrinsic declassification policies can naturally be encoded with existential types.

The contributions of this work are:
\begin{itemize}
	\item We explore an alternative type abstraction mechanism to realize the labels-as-types approach to expressive declassification, retaining the practical aspect of using an existing mechanism (here, existential types), while supporting more expressive declassification policies (Section~\ref{sec:elsec-overview}). 
	\item We define a new version of type-based relaxed noninterference, called existential relaxed noninterference, which accounts for extrinsic declassification using existential types (Section~\ref{sec:elsec-erni-overview}).
	\item We capture the essence of the use of existential types for relaxed noninterference in a core functional language \elsec 
	(Section~\ref{sec:elsec-model}), and prove that its type system soundly enforces existential relaxed noninterference (Section~\ref{sec:elsec-sec-model}). 
\end{itemize}

Section~\ref{sec:elsec-erni-formal-illustration} explains how the formal definitions apply by revisiting an example from Section~\ref{sec:elsec-erni-overview}. Section~\ref{sec:elsec-related-work} discuses related work and Section~\ref{sec:elsec-conclusion} concludes. 

\section{Overview}
\label{sec:elsec-overview}
We now explain how to use the type abstraction mechanism of existential types to denote secrets that can be selectively declassified. 
First, we give a quick overview of existential types, with their introduction and elimination forms. Next, we develop the intuitive connection between the type abstraction of standard existential types and security typing. Then, we show that to support computing with secrets, which is natural for information-flow control languages, we need to introduce faceted types.

\subsection{Existential types}
An existential type $\exists X.T$ is a pair of an (abstract) type variable 
$X$ and a type $T$ where $X$ is bound; typically $T$ provides operations to create, transform and observe values of the abstract type $X$~\cite{mitchellPlotkin:toplas1888}.

For instance,  the type $\AccountStore$ below models a simplified user repository. It provides the password of a user at 
type $X$ with the function \icode{userPass} and a function \icode{verifyPass} to check (observe) whether an arbitrary string value is equal
to the password.
\begin{displaymath}
\begin{array}{rcl}
	\AccountStore & \triangleq & \exists X.
			\begin{block}
			  \lbrack~
				\mathsf{userPass}: \String -> X \\
				~~\mathsf{verifyPass}: \String -> X -> \Bool
				\rbrack \\
			\end{block}	
\end{array}
\end{displaymath}

Values of an existential type $\exists X.T$ take the form of a {\em package} that packs together the {\em representation} 
type for the abstract type $X$ with an implementation $v$ of the operations provided by $T$. One can think of packages as {\em modules} with signatures.

For instance, the package 
$p \defas \pack{\String}{v}{\AccountStore}$ is 
a value of type $\AccountStore$, where $\String$ is the representation type and $v$, 
defined below, is a record implementing functions \icode{userPass} and \icode{verifyPass}:
\begin{displaymath}
\begin{array}{rcl}
	v & \triangleq & 
			\begin{block}
			  \lbrack~
				\mathsf{userPass} = \lambdae{x}{\String}{\mathsf{userPassFromDb}(x)} \\
				\mathsf{verifyPass}= \lambda x:\String. \lambda y:\String.\;\mathsf{equal}(x,y)
				\rbrack \\
			\end{block}	
\end{array}
\end{displaymath}
Note that the implementation, $v$, directly uses the representation type $\String$, \eg~\icode{userPass} has type $\String -> \String$ and is implemented using a primitive function $\mathsf{userPassFromDb}: \String -> \String$ 
to retrieve the user password from a database. Likewise, the implementation of \icode{verifyPass} uses equality between its arguments of type $\String$.

To use an existential type, we have to {\em open} the package (\ie~import the module) to get access to the implementation $v$, along with the abstract type that hides the actual representation type.
The expression $\unpack{X}{x}{p}{e'}$ opens the package $p$ above, exposing the representation type abstractly as a type variable $X$, and the implementation as term variable $x$, within the scope of the body $e'$. Crucially, the expression $e'$ has no access to the representation type $\String$, therefore nothing can be done with a value of type $X$, beyond using it with the operations provided by $\AccountStore$.

\subsection{Type-based declassification policies with existential types}
\label{sec:elsec-type-based-decl-with-etypes}

We can establish an analogy between existential types and selective declassification of secrets: an existential type $\exists X. T$ exposes operations to obtain \secretWord values, at the abstract type $X$, and the operations of $T$ can be used to declassify theses secrets. 

For instance, 
$\AccountStore$ provides a \secretWord string password with the function \icode{userPass}, and the function \icode{verifyPass} 
expresses the declassification policy: ``the comparison of a secret password with a public string can be made public''.
With this point of view, concrete types such as $\Bool$ and $\String$ represent \publicWord values. A fully-secret value, \ie a secret that is not declassified, can be modeled by an existential type without any observation function for the abstract type.

We can use the declassification policy modeled with $\AccountStore$ to implement a valid well-typed login functionality. 
The \icode{login} function below is defined in a scope where the package $p$ of type $\AccountStore$ is opened, providing the type name $X$ for the abstract type and the variable \icode{store} for the package implementation.

\begin{lstlisting}[numbers=none]
(*$\unpack{X}{\textsf{store}}{p}$*)
  ...
  (*$\String$*) login((*$\String$*) guess, (*$\String$*) username){
    if(store.verifyPass(guess,store.userPass(username)))
      ...
  }
\end{lstlisting}
The \icode{login} function first obtains the user \secretWord password of type $X$ with \\* \icode{store.userPass(username)}, and 
then passes the \secretWord password (of type $X$) to the function \icode{verifyPass} with the \icode{guess} \publicWord password to obtain
the \publicWord boolean result. The above code makes a valid use of $\AccountStore$ and therefore is well-typed.

The type abstraction provided by $\AccountStore$ avoids leaking information accidentally. For instance, directly returning the secret password of type $X$ is a type error, even though internally it is a string. 
Likewise, 
the expression 
\icode{length(store.userPass(username))} is ill-typed.

Note that because declassification relies on the abstraction mechanism of existential types, we work under the assumption that the person that writes the security policy---the package implementation and the existential type---is responsible for not leaking the secret due to a bad implementation or specification (\eg~exposing the secret password through the identity function of type $X->\String$).

\paragraph{Progressive declassification.}
The analogy of existential types as a mechanism to express declassification holds when one considers {\em progressive declassification}~\cite{cruzAl:ecoop2017,liZdancewic:popl2005}, which refers to the possibility of only declassifying information after a {\em sequence} of operations is performed. With existential types, we can express progressive
declassification by constraining the creation of secrets based on other secrets. 

Consider the following refinement of $\AccountStore$, which supports the declassification policy ``whether an {\em authenticated} user's salary is above \$100,000":
\begin{displaymath}
\begin{array}{rcl}
	\AccountStore & \triangleq & \exists X,Y,Z.
			\begin{block}
			  \lbrack~
				\mathsf{userPass}: \String -> X \\
				~~\mathsf{verifyPass}: \String -> X -> \optiont{Y} \\
				~~\mathsf{userSalary}: Y -> Z \\
				~~\mathsf{isSixDigit}: Z -> \Bool \rbrack
			\end{block}	
\end{array}
\end{displaymath}

$\AccountStore$ provides extra abstract types $Y$ and $Z$, denoting an authentication token (for a specific user) and a user salary, respectively. The type signatures enforce that, to obtain the user salary, the user must be authenticated: a value of type $Y$ is needed to apply \icode{userSalary}. Such a value can be obtained only after successful authentication: \icode{verifyPass} now returns an $\optiont{Y}$ value, instead of a \icode{Bool} value. 
Note that the salary itself is secret, since it has the abstract type $Z$. 
Finally, \icode{isSixDigit} function reveals whether a salary is above \$100,000 by returning a \publicWord boolean result.

Observe how the use of abstract type variables allows the existential type to enforce sequencing among operations. Also, we can provide more declassification policies for a user salary $Z$, and can use the authentication token $Y$ with more operations. An existential type is therefore an expressive means to capture rich declassification policies, including sequencing and alternation.

\subsection{Computing with secrets}

As we have seen, with standard existential types, values of an abstract type $X$ must be eliminated 
with operations provided by the existential type. 
While so far the analogy between type abstraction with existential types and expressive declassification holds nicely, there are some obstacles. 

First, with standard existential types, it is simply forbidden to compute with secrets. For instance, applying the function \icode{length}$:\String -> \Int$ with a (secret) value of type $X$ is a type error.
However, information-flow type systems are more flexible: they support computing with secret values, as long as the computation itself is henceforth considered secret, \eg~the value it produces is itself secret~\cite{zdancewic}. Allowing secret computations is useful for privileged observers, which are authorized to see private values.

{\em Faceted types} were introduced to support this ``dual mode'' of information-flow type systems in the labels-as-types approach~\cite{cruzAl:ecoop2017}. While that work is based on objects and subtyping, here we develop the notion of {\em existential faceted types}: faceted types of the form $\twoTypes{T}{U}$, where $T$ indicates the \leftFacet type used for the implementation and $U$ the \rightFacet type used for confidentiality. 

\begin{figure}[t]
\begin{displaymath}
\begin{array}{rcl}
	\AccountStore & \triangleq & \exists X, Y, Z.
			\begin{block}
			  \lbrack~
				\mathsf{userPass}: \stypebot{\String} -> \twoTypes{\String}{X} \\
				~~\mathsf{verifyPass}: \stypebot{\String} -> \twoTypes{\String}{X} -> \stypebot{\optiont{\stypebot{Y}}} \\
				~~\mathsf{userSalary}: \stypebot{Y} -> \twoTypes{\Int}{Z} \\
				~~\mathsf{isSixDigit}: \twoTypes{\Int}{Z} -> \stypebot{\Bool} \rbrack 
			\end{block}	
\end{array}
\end{displaymath}
\caption{Account store with faceted types}
\label{fig:facetedstore}
\end{figure}
Figure~\ref{fig:facetedstore} shows $\AccountStore$ with existential faceted types. 
Given a public string ($\stypebot{T}$ denotes $\twoTypes{T}{T}$), \icode{userPass} returns a value that is a string {\em for the privileged observer}, and a secret of type $X$ {\em for the public observer} (\ie~the attacker).

When computing with a value of type $\twoTypes{\String}{X}$, there are now two options: either we use a function that expects a value of type $\twoTypes{\String}{X}$ as argument, such as \icode{verifyPass}, or we use a function that goes beyond declassification, such as \icode{length}, and should therefore produce a {\em fully} private result. What type should such private results have? In order to avoid having to introduce a fresh type variable, we assume a fixed (unusable) type $\top$, and write $\stypetop{\Int}$ to denote $\twoTypes{\Int}{\top}$.

This supports computing with secrets as follows:
\begin{lstlisting}[numbers=none]
(*$\stypetop{\String}$*) login((*$\stypebot{\String}$*) guess, (*$\stypebot{\String}$*) username){
  if(length(store.userPass(username)) == length(guess))...;
}
\end{lstlisting}
Instead of being ill-typed, 
\icode{length(store.userPass(username))} is well-typed at 
type $\stypetop{\Int}$, so the function \icode{login} can return a private result, \eg  a private string at type $\stypetop{\String}$.

\subsection{Public data as (declassifiable) secret}
\label{sec:elsec-public-data-as-declassificable-secret}

Information-flow type systems allow any value to be considered private. 
With existential faceted types, this feature is captured by a subtyping relation such that for any $T$,
$\twoTypes{T}{T} <: \twoTypes{T}{\top}$, and for any $X$, 
$\twoTypes{T}{X} <: \twoTypes{T}{\top}$. Value flows that are justified by subtyping are safe from a confidentiality point of view. 
In particular, if a (declassifiable) value of type $\twoTypes{\String}{X}$
is passed at type $\stypetop{\String}$, it is henceforth fully private,  disallowing any further declassification.

Additionally, in the presence of declassifiable secrets, of type $\twoTypes{T}{X}$, one would also expect public values to be ``upgraded'' to declassifiable secrets. This requires the security subtyping relation to admit that, for any type $T$ and type variable $X$, we have 
$\stypebot{T}<:\twoTypes{T}{X}$.

Note that admitting such flows means that type variables in a \rightFacet type position are more permissive than when they occur in a safety type position. For instance, \icode{isSixDigit} can be applied to any public integer (of type $\stypebot{\Int}$), and not only to ones returned by \icode{userSalary}. In contrast, \icode{userSalary} can only be applied to a value opaquely obtained as a result of \icode{verifyPass}.
In effect, the representation of authentication tokens is still kept abstract, at type $\stypebot{Y}$ (\ie~$\twoTypes{Y}{Y}$). This prevents clients from actually knowing how these tokens are implemented, preserving the benefits of standard existential types. Conversely, the salary and password expose their representation types ($\Int$ and $\String$ respectively), thereby enabling secret computation by clients.

\section{Relaxed noninterference with existential types}
\label{sec:elsec-erni-overview}
%

Existential faceted types support a novel notion of type-based relaxed noninterference called {\em existential 
relaxed noninterference} (ERNI) that defines if a program with existential faceted types is secure. 
ERNI is based on type-based {\em equivalences} between values at existential faceted types. We formally define the notions of type-based equivalence and ERNI in Section~\ref{sec:elsec-sec-model}, but here we provide an intuition for this security criterion and the associated reasoning.
Let us first consider simple types, before looking at existential types.

\paragraph{Type-based relaxed noninterference.}
Two integers are equivalent at type $\twoTypes{\Int}{\Int} = \stypebot{\Int}$ if they are syntactically equal, meaning that 
a public observer can distinguish between two integers at type $\stypebot{\Int}$. We can characterize the meaning of the faceted type $\stypebot{\Int}$ with the partial equivalence relation $Eq_{\Int} =  \{(\textsf{n},\textsf{n})\in \Int \times \Int\}$. Using this, two integers 
$\mathsf{v_1}$ and $\mathsf{v_2}$ are equivalent at type $\stypebot{\Int}$ if they are in the relation $Eq_{\Int}$---meaning they are syntactically equal.

Dually, the type $\twoTypes{\Int}{\top}= \stypetop{\Int}$ characterizes 
integer values that are indistinguishable for a public observer, therefore 
{\em any} two integers are equivalent at type $\stypetop{\Int}$.
Consequently, the meaning of the faceted type $\stypetop{\Int}$ is the total relation $All_{\Int} = \Int \times \Int$ that relates any two integers $\mathsf{v_1}$ and $\mathsf{v_2}$.

With these base type-based equivalences, one can express the security property of functions, open terms, and programs with inputs as follows: 
a program $p$ satisfies ERNI at an observation type $S_{out}$ if, given two input values that are equivalent at type $S_{in}$, the executions of $p$ with each value produce results that are equivalent at type $S_{out}$. This modular reasoning principle is akin to standard noninterference~\cite{zdancewic} and type-based relaxed noninterference~\cite{cruzAl:ecoop2017}.

Intuitively, $S_{in}$ models the {\em initial knowledge} of the public observer about the (potentially-secret) input, and $S_{out}$ denotes the {\em final knowledge} that the public observer has to distinguish results of the executions of the program $p$. The program $p$ is secure if, given inputs from the same equivalence class of $S_{in}$, it produces results in the same equivalence class of $S_{out}$.
Consider the program $e = \mathsf{length(x)}$ where \icode{x} has type $\stypetop{\String}$. The program $e$ does not satisfy ERNI at type $\stypebot{\Int}$, because given two strings 
\icode{``a"} and \icode{``aa"} that are equivalent at $\stypetop{\String}$, \ie $(\mathsf{``a"},\mathsf{``aa"}) \in All_{\String}$, we obtain 
the results \icode{1} and \icode{2}, which are not equivalent at type $\stypebot{\Int}$, \ie $(\mathsf{1},\mathsf{2}) \notin Eq_{\Int}$. However,
$e$ is secure at type $\stypetop{\Int}$.

\paragraph{Relaxed noninterference and existentials.}
When we introduce faceted types with type variables such as $\twoTypes{\Int}{X}$, we need to answer: what values are equivalent at type $\twoTypes{\Int}{X}$?
Without stepping into technical details yet, let us say that the meaning of a type $\twoTypes{\Int}{X}$ is an {\em arbitrary} partial equivalence relation $R_{X} \subseteq \Int \times \Int$, and two values $\mathsf{v_1}$ and $\mathsf{v_2}$ are equivalent at $\twoTypes{\Int}{X}$ if they are in $R_{X}$. 
Because $X$ is an existentially-quantified variable, inside the package implementation that exports the type variable $X$, $R_{X}$ is known, 
but outside the package, \ie for clients of 
a type $\twoTypes{\Int}{X}$, $R_{X}$ is completely abstract: a public observer that opens a package exporting the type variable $X$ does not know anything about values of type $\twoTypes{\Int}{X}$.
For instance, consider again the program $e= \mathsf{length(x)}$ but assume that  \icode{x} now has type $\twoTypes{\String}{Y}$. Does $e$ satisfy ERNI at type $\stypebot{\Int}$? 
Here, we need to know what is the relation $R_{Y}$ that gives meaning to $\twoTypes{\String}{Y}$. 
Instead of picking only one relation $R_{Y}$, ERNI quantifies over {\em all} possible relations $R_{Y}$.
This universal quantification over $R_{Y}$ corresponds to the standard type abstraction mechanism for abstract types (\ie parametricity~\cite{reynolds:83}). That is, the program $e$ satisfies ERNI at type $\twoTypes{\Int}{\Int}$, if it is secure for all relations $R_{Y} \subseteq \String \times \String$. 
Then, to show that ERNI at type $\twoTypes{\Int}{\Int}$ does not hold for $e$ it suffices to exhibit a specific relation for which ERNI is violated. Take the relation $R_{Y} = \{(\mathsf{``a"},\mathsf{``aa"})\}$, and observe that \icode{length(``a")} $\neq$ \icode{length(``aa")}.

\paragraph{Illustration.} Finally, we give an intuition of how ERNI accounts for extrinsic declassification policies.
We reuse the salary operations from $\AccountStore$, simplifying the retrieval of the secret salary. The type $\SalaryPolicy$ provides
a secret \icode{salary} and a function  \icode{isSixDigit} to declassify the salary as before.
\begin{displaymath}
\SalaryPolicy \defas \exists Z. \lbrack \mathsf{salary}: \twoTypes{\Int}{Z},~ \mathsf{isSixDigit}: \twoTypes{\Int}{Z} -> \stypebot{\Bool} \rbrack
\end{displaymath}

Intuitively two package values $\mathsf{p_1} \defas$ $\pack{\Int}{\mathsf{v_1}}{\SalaryPolicy}$ and $\mathsf{p_2} \defas$ $\pack{\Int}{\mathsf{v_2}}{\SalaryPolicy}$ are equivalent at $\stypebot{\SalaryPolicy}$ if $\mathsf{v_1.salary}$ and $\mathsf{v_2.salary}$ 
are equivalent at $\twoTypes{\Int}{Z}$ and $\mathsf{v_1.isSixDigit}$ and $\mathsf{v_2.isSixDigit}$ are equivalent at $\twoTypes{\Int}{Z} -> \stypebot{\Bool}$, 
\ie the functions \icode{isSixDigit} of $\mathsf{v_1}$ and $\mathsf{v_2}$ 
produce equivalent results at $\stypebot{\Bool}$ when given equivalent arguments at $\twoTypes{\Int}{Z}$.

Consider now the program $e' = \mathsf{x.isSixDigit (x.salary)}$ with input $\mathsf{x}$ of type $\lbrack \mathsf{salary}: \twoTypes{\Int}{Z},~ \mathsf{isSixDigit}: \twoTypes{\Int}{Z} -> \stypebot{\Bool} \rbrack$.
Taking into account the above equivalence for $\SalaryPolicy$, the program $e'$
satisfies ERNI at type $\stypebot{\Bool}$ because it adheres to the salary policy.
Indeed, given {\em any} two equivalent packages $\mathsf{p_1} \defas$ $\pack{\Int}{\mathsf{v_1}}{\SalaryPolicy}$ and $\mathsf{p_2} \defas$ $\pack{\Int}{\mathsf{v_2}}{\SalaryPolicy}$, the expressions    
$\mathsf{v_1.isSixDigit (v_1.salary)}$ and $\mathsf{v_2.isSixDigit (v_2.salary)}$ are {\em necessarily} equivalent at type $\stypebot{\Bool}$.
However, the program $\mathsf{(x.salary) \% 2}$ does not satisfy ERNI at $\stypebot{\Int}$, because given equivalent package implementations $\mathsf{v_1} \triangleq \lbrack \mathsf{salary} = \mathsf{100001}, \cdots \rbrack$ and 
$\mathsf{v_2} \triangleq \lbrack \mathsf{salary} = \mathsf{100002}, \cdots \rbrack$, it yields \icode{1} for $\mathsf{v_1}$ and \icode{0} for $\mathsf{v_2}$, and both values are not equivalent at $\stypebot{\Int}$, \ie $(\mathsf{1},\mathsf{0}) \notin Eq_{\Int}$.

\section{Formal semantics}
\label{sec:elsec-model}
We model existential faceted types in \elsec, 
which is essentially the simply-typed lambda calculus augmented with the unit type, pair types, sum types, existential types, and faceted types. All the examples presented in Section~\ref{sec:elsec-overview} can thus be encoded in \elsec using standard techniques.
This section covers the syntax, static and dynamic semantics of \elsec. The formalization of existential relaxed noninterference and the security type soundness of \elsec are presented in Section~\ref{sec:elsec-sec-model}.

\subsection{Syntax}
Figure~\ref{fig:elsec-syntax} presents the syntax of \elsec. Expressions $e$ are completely standard~\cite{pierce:tapl}, including functions, applications, variables, primitive values
, binary operations on primitive values, the unit value, pairs with their first and second projections, injections $\inl{e}$ and $\inr{e}$ to introduce sum types, as well as a \icode{case}
construct to eliminate sums; finally, \icode{pack} and \icode{open} introduce and eliminate existential packages, respectively.
Types $T$ include
function types $S->S$, primitive types $\primt$, the unit type $\unittype$, sum types $\sumt{S}{S}$, pair types 
$\pairt{S}{S}$, existential types $\eType{X}{T}$, type variables $X$ and the top type $\top$. A security type $S$ 
is a faceted type $\twoTypes{T}{U}$ where $T$ is the \leftFacet type and $U$ is the \rightFacet type.

\begin{figure}[t]
  \begin{small}
		\begin{displaymath}
				\begin{array}{rcll}
					e & ::= & 
					\begin{block}
					\lambda x:S.e |  \appe{e}{e} | x | \primb | \binOp{e}{e} | \unitval | \paire{e}{e} \\
					| \pFst{e} | \pSnd{e} | \inl{e} | \inr{e} | \sCase{t}{x}{e}{x}{e} \\
					| \pack{T}{e}{T} | \unpack{X}{x}{e}{e}
					\end{block} & \text{(terms)}\\
					v & ::= & \lambda x:S.e | \primb | \unitval| \paire{v}{v} | \inl{v} | \inr{v} |  \pack{T}{v}{T} |  & \text{(values)}\\										
					T, U& ::= & S -> S | \primt | \unittype | \sumt{S}{S} | \pairt{S}{S} | \eType{X}{T} | X | \top & \text{(types)}\\
					\primt & ::= & (\eg~\Int, \String) & \text{(primitive types)}\\
					S & ::= & \twoTypes{T}{U} & \text{(security types)}\\
				\stypebot{T} &\triangleq & \twoTypes{T}{T} \qquad \stypetop{T} \triangleq \twoTypes{T}{\top}
				\end{array}
			\end{displaymath}
 \end{small}
\caption{\elsec: Syntax} 
\label{fig:elsec-syntax}
\end{figure}

\paragraph{Well-formedness of security types.}
We now comment on the rules for valid security types, \ie {\em facet-wise} well-formed types. We have three general form of security types $\stypebot{T}$ and $\stypetop{T}$ and 
$\twoTypes{T}{X}$. While there is no constraint on forming types $\stypebot{T}$ and $\stypetop{T}$, such as $\stypebot{\Int}$, $\stypebot{X}$ and $\stypetop{\Int}$, we need two considerations for types such 
as $\twoTypes{\Int}{X}$.

The first consideration is that {\em inside} an existential type $\exists X.T$ 
the type variable $X$, when used as a declassification type, must be uniquely associated to a concrete safety type. For instance, the existential type $\eType{X}{(\twoTypes{\String}{X} -> \twoTypes{\Int}{X})}$ is ill-formed, while $\eType{X}{(\twoTypes{\Int}{X} -> \twoTypes{\Int}{X})}$ and $\eType{X}{(\twoTypes{X}{X} -> \twoTypes{\Int}{X})}$ are well-formed. 
For such well-formed types, we use the auxiliary function $\repType{\eType{X}{T}}$ to obtain the safety type associated to $X$; for instance $\repType{\eType{X}{(\twoTypes{X}{X} -> \twoTypes{\Int}{X})}} = \Int$ (undefined on ill-formed types).

The second consideration is when a client opens a package.
The expression $\unpack{X}{x}{e}{e'}$ binds the type variable $X$ in $e'$, therefore the expression $e'$ can declare security types of the form 
$\twoTypes{T'}{X}$. However, for the declaration of the type $\twoTypes{T'}{X}$ to be valid, the safety type of the declassification type variable $X$ must be $T'$.
For instance, if the safety type of $X$ is $\Int$, the expression $e'$ cannot declare security types such as $\twoTypes{\String}{X}$, otherwise computations over secrets could get stuck.
The question is how to determine the safety type $T'$ of $X$ in $e'$. 
Crucially, the expression $e$ necessarily has to be of type $\stypebot{(\exists X.T)}$, therefore we can obtain the safety type
for $X$ with $\repType{\eType{X}{T}}$. To keep track of the safety type for each type variable $X$, we use a type variable environment $\Delta$ that maps type variables to types $T$ 
(\ie $\Delta ::= \bigcdot | \Delta, X:T$)%
With the previous considerations in mind, the rules for well-formed security types are straightforward. 
In the rest of the paper, we use the judgment $\wf{\Delta}{S}$ to mean {\em well-formed} security types $S$ under type environment $\Delta$. A well formed security type $S$ is 
both facet-wise well-formed and closed with respect to type variables. We also use $\Delta \models \Gamma$ to indicate that a type environment is well-formed, \ie all types in $\Gamma$ 
are well-formed. 
In the following, we assume well-formed security types and environments.

\subsection{Static Semantics}

Figure~\ref{fig:elsec-static-semantics} presents the static semantics 
of \elsec. Security typing relies on a subtyping judgment that validates secure information flows. 
The left-most rule justifies subtyping by reflexivity. The middle rule justifies subtyping for two security types with 
the same \leftFacet type, when the \rightFacet type of right security type is $\top$. Finally, the right-most rule 
justifies subtyping between a \publicWord type $\stypebot{T}$ and
$\twoTypes{T}{X}$.

As usual, the typing judgment $\Delta; \Gamma |- e : S$ 
denotes that ``the expression e has type $S$ under the type variable environment $\Delta$ and the type environment $\Gamma$''. 
The typing rules are mostly standard~\cite{pierce:tapl}. 
Here, we only discuss the special treatment of security types.

\begin{figure}[!htbp]
\begin{small}
\framebox{$S <: S$}
\begin{mathpar}	
		\twoTypes{T}{U} <: \twoTypes{T}{U}
	\qquad
      \twoTypes{T}{U} <: \twoTypes{T}{\top}
	\qquad
      \stypebot{T} <: \twoTypes{T}{X}
\end{mathpar}
\framebox{$\Delta; \Gamma |- e : S$}
\begin{mathpar}
	\inference[(TVar)]{
				x \in dom(\Gamma)
			}{
				\Delta; \Gamma |- x: \Gamma(x)
			}
			\quad
	\inference[(TS)]{
		\Delta; \Gamma |- e: S' & S' <: S
	}{
		\Delta; \Gamma |- e: S
	}
	\quad
	\inference[(TP)]{
			\primt = \Theta(\primb)
		}{
		\DeltaX;\Gamma |- \primb: \stypebot{\primt}
	} 
	\\
  	\inference[(TFun)]{			
			\Delta;\Gamma, x: S |- e : S'
    }{
      \Delta;\Gamma |- \lambdae{x}{S}{e}:  \stypebot{(S -> S')}
	}
	\quad
	\inference[(TPair)]{			
		\Delta;\Gamma |- e_1: S_1 \quad
		\Delta;\Gamma |- e_2: S_2 
    }{
      \Delta;\Gamma |- \paire{e_1}{e_2}: \stypebot{(\pairt{S_1}{S_2})}
	}	
	\\
	\inference[(TU)]{
		~
		}{
		\DeltaX;\Gamma |- \unitval{}: \stypebot{\unittype}
	}
	\quad 
	\inference[(TInl)]{			
			\Delta;\Gamma |- e: S_1 
	}{
	\Delta;\Gamma |- \inl{e}: \stypebot{(\sumt{S_1}{S_2})}
	} \quad
	\inference[(TInr)]{			
			\Delta;\Gamma |- e: S_2 
	}{
	\Delta;\Gamma |- \inr{e}: \stypebot{(\sumt{S_1}{S_2})}
	}
	\\
  	\inference[(TPack)]{
      \Delta;\Gamma |- e: \stypebot{(\ssubst{T}{T'}{X})}\\ 
			\Gbox{\precise{T'}{\repType{\eType{X}{T}}}}
    }{
      \Delta;\Gamma |- \pack{T'}{e}{\eType{X}{T}}: \stypebot{(\eType{X}{T})}
	}
	\quad
	\inference[(TApp)]{
	  \Delta;\Gamma |- e_1 : S & S = \twoTypes{(\mtype{S_1}{S_2})}{U} \\
		\Delta;\Gamma |- e_2: S_1
	}{
	\Delta;\Gamma |- \appe{e_1}{e_2}: \stamp{S_2}{S}
	}
	\\	
	\inference[(TOp)]{
		\Theta(\binOpSym): \primt \times \primt' -> \primt'' &
		\DeltaX;\Gamma |- e_1 : \twoTypes{\primt}{U} &
		\DeltaX;\Gamma |- e_2 : \twoTypes{\primt'}{U'}
	}{
		\DeltaX;\Gamma |- \binOp{e_1}{e_2}: \stamp{\stamp{\twoTypes{\primt''}{\primt''}}{\twoTypes{\primt}{U}}}{\twoTypes{\primt'}{U'}} 
	}\\
	\inference[(TFst)]{
		\Delta;\Gamma |- e : S \\ S = \twoTypes{(\pairt{S_1}{S_2})}{U}
	}
	{
		\Delta;\Gamma |- \pFst{e}: \stamp{S_1}{S}
	}\quad	
		\inference[(TSnd)]{
		\Delta;\Gamma |- e : S \\
		S = \twoTypes{(\pairt{S_1}{S_2})}{U}
	}{
		\Delta;\Gamma |- \pSnd{e}: \stamp{S_2}{S}
	}
	\\
	\inference[(TCase)]{
	\Delta;\Gamma |- e : S & S = \twoTypes{(\sumt{S_1}{S_2})}{U}\\
	\Delta;\Gamma,x_1 : S_1 |- e_1: S' &
	\Delta;\Gamma,x_2 : S_2 |- e_2: S' 
	}{
	\Delta;\Gamma |- \sCase{e}{x_1}{e_1}{x_2}{e_2}: \stamp{S'}{S}
	} \\
	\inference[(TOpen)]{
		\Delta;\Gamma |- e: S & S = \twoTypes{(\eType{X}{T})}{U} & \Gbox{T' \defas \repType{\eType{X}{T}}}\\
		\Delta,X:T';\Gamma, x: \stypebot{T} |- e': S'  & \wf{\Delta}{S'}
	}{
		\Delta;\Gamma |- \unpack{X}{x}{e}{e'}: \stamp{S'}{S}
	}
\end{mathpar}
\end{small}
\caption{\elsec: Static semantics}
  \label{fig:elsec-static-semantics}
\end{figure}

Rule (TVar) gives the security type to a type variable from the type environment and rule (TS) is the standard subtyping subsumption rule. 
Rules (TP), (TFun), (TPair), (TU), (TInl), (TInr) and (TPack) introduce primitive, function, pair, unit, sum and existential types, respectively.
In particular, rule (TPack)  requires the representation type of the package to 
be {\em more precise} than the safety type associated to $X$ in the existential type, \ie $\precise{T'}{\repType{\eType{X}{T}}}$. The precision judgment has only two rules: reflexivity $\precise{T}{T}$, and any type is more precise than a type variable $\precise{T}{X}$.

Rules (TApp), (TOp), (TFst), (TSnd), (TCase) and (TOpen) are elimination rules for function, primitive, pair, sum and existential types, respectively. 
When a secret is eliminated, the resulting computation must protect that secret. This is done with $\stamp{S'}{S}$, which changes the \rightFacet type of $S'$ to $\top$ if the type $S$ is not public:
$$\stamp{\twoTypes{T_1}{U_1}}{\twoTypes{T_2}{U_2}} = \twoTypes{T_1}{U_1} 
\text{ if }  T_2 = U_2, \text{ otherwise } \twoTypes{T_1}{\top}$$

Let us illustrate the use of rule (TApp). 
On the one hand, if the type of the function expression $e_1$ is $\stypebot{(S_1 -> S_2)}$, \ie it represents a public function, then
the type of the function application is $S_2$. On the other hand, if the function expression $e_1$ has type $\twoTypes{(S_1 -> \twoTypes{T_2}{U_2})}{X}$ or $\stypetop{(S_1 -> \twoTypes{T_2}{U_2})}$, \ie it represents a secret, then the function application has type $\twoTypes{T_2}{\top}$.

Rule (TOp) uses an auxiliary function $\Theta$ to obtain the signature of a primitive operator and ensures that the resulting type protects both operands with $\stamp{\stamp{\twoTypes{\primt''}{\primt''}}{\twoTypes{\primt}{U}}}{\twoTypes{\primt'}{U'}} $. Rules (TFst) and (TSnd) use the same principle to protect the projections of a pair.
Rule (TCase) requires the discriminee to be of type $\twoTypes{(\sumt{S_1}{S_2})}{U}$, and both branches must have the same type $S'$. Likewise, it protects the resulting computation with $\stamp{S'}{S}$. 

Finally, rule (TOpen) applies to expressions of the form $\unpack{X}{x}{e}{e'}$, by typing the body expression $e'$ in an 
extended type variable environment $\Delta,X:T'$ and a type environment $\Gamma, x:\stypebot{T}$. Two points are worth noticing.
First, the association $X:T'$ allows us
to verify that security types of the form $\twoTypes{T'}{X}$ defined in the body expression $e'$ are well-formed. 
Second, we make the well-formedness requirement explicit for the result type $\wf{\Delta}{S'}$, which implies that $S'$ is facet-wise well-formed and closed under $\Delta$---\ie the type variable $X$ cannot appear in $S'$.

\subsection{Dynamic semantics and type safety}

The execution of \elsec expressions is defined with a standard call-by-value small-step dynamic semantics based 
on evaluation contexts (Figure~\ref{fig:elsec-dynamic-semantics}). 
We abstract over the execution of primitive operators over primitive values using an auxiliary function $\theta$.
\begin{figure}[h]
\begin{small}	
	\begin{mathpar}			
			\begin{array}{llll}
				E & ::= & 
				\begin{block}
					\left[~\right] | \pFst{E} | \pSnd{E} | \sCase{E}{x_1}{e_1}{x_2}{e_2} \\
					| \appe{E}{e} | \appe{v}{E} 
						| \binOp{E}{e} | \binOp{v}{E} | \unpack{X}{x}{E}{e'} 
				\end{block}
				& \text{(evaluation contexts)}
			\end{array}
			\\
			\begin{block}
			\pFst{\paire{v_1}{v_2}} \reduce v_1 \\
			\pSnd{\paire{v_1}{v_2}} \reduce v_2 \\
			\sCase{\inl{v}}{x_1}{e_1}{x_2}{e_2} \reduce \ssubst{e_1}{v}{x_1} \\
			\sCase{\inr{v}}{x_1}{e_1}{x_2}{e_2} \reduce \ssubst{e_2}{v}{x_2} \\
			\appe{(\lambdae{x}{S}{e})}{v} \reduce \ssubst{e}{v}{x} \\
			\binOp{\primb_1}{\primb_2} \reduce \theta(\binOpSym,\primb_1,\primb_2) \\
			\unpack{X}{x}{(\pack{T'}{v}{\eType{X}{T}})}{e'} \reduce \ssubst{\ssubst{e'}{v}{x}}{T'}{X}
			\end{block}
			\inference{
				e \reduce e'
			}{
				E[e] \reduce E[e']
			}			
	\end{mathpar} 
\end{small}
\caption{\elsec: Dynamic semantics}
\label{fig:elsec-dynamic-semantics}
\end{figure}


We define the predicate $\mathsf{safe}(e)$ to indicate that the evaluation of the expression $e$ does not get stuck. 
\label{sec:elsec-safety}
\begin{restatable}[Safety]{definition}{elsecSafe}
\label{def:elsec-safe}
$\mathsf{safe}(e) \Longleftrightarrow \forall e'.~ e \reduce^{*} e' \implies e' = v~ or ~\exists e''.~ e' \reduce e''$
\end{restatable}
Well-typed \elsec closed terms are safe.
\begin{restatable}[Syntactic type safety]{theorem}{gobsecTypeSafety}
\label{the:gobsec-type-safety}
$ |- e : S  \implies  \mathsf{safe}(e)$
\end{restatable}

Having formally defined the language \elsec, we move to the main result of this paper, which is to show that the \elsec is sound from a security standpoint, \ie its type system enforces existential relaxed noninterference.

\section{Existential relaxed noninterference, formally}
\label{sec:elsec-sec-model}

In Section~\ref{sec:elsec-erni-overview} we gave an overview of existential relaxed noninterference (ERNI), explaining how it depends on type-based equivalences. To formally capture these type-based equivalences, we define a logical relation, defined by induction on the structure of types. To account for type variables, we build upon prior work on logical relations for parametricity~\cite{ahmed:esop2006,reynolds:83,wadler:fpca89}. Then, we formally define ERNI on top of this logical relation. Finally, we prove that the type system of \elsec enforces existential relaxed noninterference.

\subsection{Logical relation for type-based equivalence}

As explained in Section~\ref{sec:elsec-erni-overview}, two values $v_1$ and $v_2$ are equivalent at type $S$, if they are
in the partial equivalence relation denoted by $S$. To capture this, 
the logical relation (Figure~\ref{fig:elsec-logical-relation}) interprets types as set of {\em atoms}, \ie pairs of closed expressions. 
We use $\eatomtwo{T_1}{T_2}$ 
to characterize the set of atoms with expressions of type $T_1$ and $T_2$ respectively. This definition appeals to a simply-typed judgment
$\stypeof{\Delta;\Gamma}{e}{T}$ that does not consider the \rightFacet type and is therefore completely standard. 
The use of this simple type system clearly separates the {\em definition} of secure programs from the {\em enforcement} mechanism, \ie the security type system of Figure~\ref{fig:elsec-static-semantics}.

In Section~\ref{sec:elsec-erni-overview} we explained what it means to be equivalent at type $\twoTypes{\Int}{X}$ appealing to a relation on integers $R_{X} \subseteq 
\Int \times \Int$ . To formally characterize the set of valid relations $R_{X}$ %
 for types $T_1$ and $T_2$ we use the definition $\erel{T_1}{T_2}$. To keep track of the relation associated to a type variable, most definitions are 
indexed by an environment $\rho$ that maps type variables $X$ to triplets $(T_1,T_2,\XRel)$, where $T_1$ and $T_2$ are  
two representation types of $X$  and $\XRel$ is a relation on closed values of type $T_1$ and $T_2$ (\ie $\rho ::= \emptyset | \xsrho{\rho}{X}{(T_1,T_2,\XRel)}$). We will explain
later where these types $T_1$ and $T_2$ come from.
We write $\rhofst(U)$ (resp. $\rhosnd(U)$) to replace all type variables of $\rho$ in types with the associated type $T_1$ (resp. $T_2$), 
and $\rhosem{X}$ to retrieve the relation $\XRel$ of a type variable $X$ in $\rho$.

Figure~\ref{fig:elsec-logical-relation} defines the 
{\em value interpretation} of a type $T$, denoted $\esetv{T}$, 
then the value interpretation of a security type $S$, denoted $\esetv{S}$, and finally the {\em expression interpretation} of a type $S$, denoted $\gsetc{S}$.

\paragraph{Interpreting concrete types.}
We first explain the definitions that do not involve types variables. 
$\esetv{\twoTypes{T}{\top}}$ (resp. $\esetv{\twoTypes{T}{T}}$) characterizes
when values of $T$ are indistinguishable (resp. distinguishable) for the public observer. 
$\esetv{\twoTypes{T}{\top}}$ is defined as $\eatomtwo{\rhofst(T)}{\rhosnd(T)}$ indicating that any two values of type $T$ are equivalent at type $\twoTypes{T}{\top}$. Note that this also includes values of type $\twoTypes{X}{\top}$.

Two {\em public} values are equivalent at a security type $\twoTypes{T}{T}$ if they are equivalent at their safety type, \ie~$\esetv{\stypebot{T}} = \esetv{T}$.
The definition $\esetv{\primt}$ relates syntactically-equal primitive values at type $\primt$. Two functions are equivalent at type $S_1 -> S_2$, denoted $\esetv{S_1 -> S_2}$, if given equivalent arguments at type $S_1$, 
their applications are equivalent expressions at type $S_2$. Two pairs are equivalent at type $\pairt{S_1}{S_2}$ if they are component-wise equivalent. Two values are equivalent at $\sumt{S_1}{S_2}$ if they are either both left-injected values $\inl{v_1}$ and $\inl{v_2}$ such as $v_1$ and $v_2$ are equivalent at $S_1$, or both right-injected values $\inr{v_1}$ and $\inr{v_2}$ such as $v_1$ and $v_2$ are equivalent at $S_2$.

Finally, two expressions $e_1$ and $e_2$ are equivalent at type $\twoTypes{T}{U}$, denoted $\gsetc{\twoTypes{T}{U}}$, if they both reduce to values $v_1$ and $v_2$ respectively and these values are related at type $\twoTypes{T}{U}$. 
(Note that all well-typed \elsec expressions terminate.)

\begin{figure}[t]
	\begin{small}
	\begin{displaymath}
		\begin{array}{lcl}
			\eatomtwo{T_1}{T_2} & = & \{ (e_1,e_2) |~\stypeof{\bigcdot;\bigcdot}{e_1}{T_1}~\wedge~\stypeof{\bigcdot;\bigcdot}{e_2}{T_2} \}\\
			\eatomunion{T} & = & \eatomtwo{\rhofst(T)}{\rhosnd(T)}\\
			\erel{T_1}{T_2} &  = & \{\XRel \subseteq \eatomtwo{T_1}{T_2}\} \\
			\esetv{\unittype}  &  = & \{(\unitval,\unitval) \in \eatomunion{\unittype} \}\\
			\esetv{\primt}  &  = & \{(\primb,\primb) \in \eatomunion{\primt} \}\\			
			\esetv{S_1->S_2}  &  = &
			\begin{block}
				\{ (v_1, v_2) \in \eatomunion{S_1->S_2} | 
				\\
					\forall v_1',v_2'.~(v_1', v_2') \in \esetv{S_1} \implies 
					(\appe{v_1}{v'_1}, \appe{v_2}{v'_2}) \in \gsetc{S_2}
				\}			
			\end{block} \\			
			\esetv{\pairt{S_1}{S_2}}  &  = &
			\begin{block}
				\{ (\paire{v_1}{v'_1}, \paire{v_2}{v'_2}) \in \eatomunion{\pairt{S_1}{S_2}} |~
				  (v_1,v_2) \in \esetv{S_1}~\wedge~ (v'_1,v'_2) \in \esetv{S_2}
				\}
			\end{block} \\			
			\esetv{\sumt{S_1}{S_2}}  &  = &
			\begin{block}
				\{ (\inl{v_1}, \inl{v_2}) \in \eatomunion{\sumt{S_1}{S_2}} |~
				  (v_1,v_2) \in \esetv{S_1}
				\}
				\\
				\union
				\{ (\inr{v_1}, \inr{v_2}) \in \eatomunion{\sumt{S_1}{S_2}} |~
				  (v_1,v_2) \in \esetv{S_2}
				\}
			\end{block} \\			
			\esetv{\eType{X}{T}}  &  = &
			\begin{block}
				\{ (\pack{T_1}{v_1}{\eType{X}{T}}, \pack{T_2}{v_2}{\eType{X}{T}}) \in \eatomunion{\eType{X}{T}} | 
				\\ \Gbox{\precise{T_1}{\repType{\eType{X}{T}}}~\wedge~\precise{T_2}{\repType{\eType{X}{T}}}} \wedge\\
				\exists \XRel \in \unionrel{T_1,T_2}. 
					\quad (v_1, v_2) \in \gsetvx{T}{\xsrho{\rho}{X}{(T_1,T_2,\XRel)}}
				\}
			\end{block} \\			
			\esetv{X}  &  = & \rhosem{X}\\
			\esetv{\twoTypes{T}{X}}  &  = & \rhosem{X} \union \esetv{T} \\
			\esetv{\twoTypes{T}{\top}}  &  = & \eatomtwo{\rhofst(T)}{\rhosnd(T)}\\
			\esetv{\twoTypes{T}{T}}  &  = & \esetv{T}\\
			\gsetc{\twoTypes{T}{U}} &  =  & 
			\begin{block}			 
				\{ (e_1, e_2) \in \eatomunion{T}|~
					\forall v_1,v_2.~ 
					e_1 \reduce^{*} v_1~\wedge~ ~ e_2 \reduce^{*} v_2~\wedge~
					(v_1,v_2) \in \esetv{\twoTypes{T}{U}}
				\}
			\end{block} \\
    \end{array}		
  \end{displaymath}	
	\end{small}
 \caption{\elsec Logical relation for type-based equivalence}
  \label{fig:elsec-logical-relation}
\end{figure}

\paragraph{Interpreting existential types.}
We now explain the value interpretation of existential types $\eType{X}{T}$, type variables $X$ and security types of the form $\twoTypes{T}{X}$, which all involve type variables.

Two public package expressions $\pack{T_1}{v_1}{\eType{X}{T}}$ and $\pack{T_2}{v_2}{\eType{X}{T}}$ are equivalent at type $\eType{X}{T}$, denoted $\esetv{\eType{X}{T}}$, 
if there exists a relation $\XRel$ on the representation types $T_1$ and $T_2$ that makes the package implementations $v_1$ and $v_2$ 
equivalent at type $T$, denoted $(v_1, v_2) \in \gsetvx{T}{\xsrho{\rho}{X}{(T_1,T_2,\XRel)}}$. Note that if the existential type $\eType{X}{T}$ has a concrete 
safety type $T'$ (not a type variable) for $X$, then both $T_1$ and $T_2$ necessarily have to be equal to $T'$. Otherwise, $T_1$ and $T_2$ are arbitrary types.
Two values are related at type $X$, denoted $\esetv{X}$, if they are in the relational interpretation $\XRel$ associated to $X$ (retrieved with 
$\rhosem{X}$). Two values are related at type $\twoTypes{T}{X}$, denoted $\esetv{\twoTypes{T}{X}}$, if they are in $\rhosem{X}$, or if they are publicly-equivalent 
values of type $T$ (\ie~a package can accept 
public values of type $T$ where values of $\twoTypes{T}{X}$ are expected).

We illustrate these formal type-based equivalences in Section~\ref{sec:elsec-erni-formal-illustration}, after formally defining existential relaxed noninterference and proving security type soundness.

\subsection{Existential relaxed noninterference}

As illustrated in Section~\ref{sec:elsec-erni-overview},
ERNI is a modular property that accounts for open expressions over both variables and type variables. To account for open expressions,
we first need to define the relational interpretation of a type environment $\Gamma$ and a type variable 
environment $\Delta$:
	\begin{small}
	\begin{displaymath}
		\begin{array}{lcl}			
			\gsetg{\cdot} & = & \{(\emptyset,\emptyset)\} \\
			\gsetg{\Gamma;x:S} & = &
				\begin{block}					
					\{(\extenv{\gamma_1}{x}{v_1},\extenv{\gamma_2}{x}{v_2})|~
					(\gamma_1,\gamma_2) \in \gsetg{\Gamma}
					~\wedge
					~ (v_1, v_2) \in \esetv{S}
					\} 
				\end{block} \\
			\gsetd{\cdot} & = & \{\emptyset\} \\
			\gsetd{\Delta;X:T} & = & \{\xsrho{\rho}{X}{(T_1,T_2,\XRel)} | \rho \in \gsetd{\Delta} ~\wedge~ \precise{T_1}{T}~\wedge~\precise{T_2}{T} ~\wedge~ \XRel \in \erel{T_1}{T_2} \}
    \end{array}		
  \end{displaymath}	
	\end{small}

 The type environment interpretation $\gsetg{\Gamma}$ is standard; it characterizes when two value substitutions 
$\gamma_1$ and $\gamma_2$ are equivalent. A value substitution $\gamma$ is a mapping from variables to closed values (\ie $\gamma ::= \emptyset | \extenv{\gamma}{x}{v}$). 
Two value substitutions are equivalent if for all associations $x:S$ in $\Gamma$, the mapped values  to $x$ in $\gamma_1$ and $\gamma_2$ are
equivalent at $S$. Finally, the interpretation of a type variable environment $\Delta$, denoted $\gsetd{\Delta}$, is a set 
of type substitutions $\rho$ with the same domain as $\Delta$. For each type variable $X$ bound to $T$ in $\Delta$, such a $\rho$ maps $X$ to triples $(T_1,T_2,\XRel)$, where $T_1$ and $T_2$ are closed types that are more precise than $T$
. $\XRel$ must be a valid relation for the types $T_1$ and $T_2$. We write $\rhofst(e)$ (resp. $\rhosnd(e)$) 
to replace all type variables of $\rho$ in terms with their associated type $T_1$ (resp. $T_2$), 

We can now formally define ERNI. An expression $e$ 
satisfies existential relaxed noninterference for a type variable environment $\Delta$ and a type variable $\Gamma$ at the $S$, denoted 
$\erni{\Delta}{\Gamma}{e}{S}$ if, given a type substitution $\rho$ satisfying $\Delta$ and two values substitutions 
$\gamma_1$ and $\gamma_2$ that are equivalent at $\Gamma$, applying the substitutions produces equivalent expressions at type $S$.

\begin{definition}[Existential relaxed noninterference]
\label{def:elsec-erni}
\begin{displaymath}			
		\begin{array}{l}
			\erni{\DeltaX}{\Gamma}{e}{S} \Longleftrightarrow \exists T,U.~S \triangleq \twoTypes{T}{U} ~\wedge~ \stypeof{\Delta;\Gamma}{e}{T}~\wedge~ \wf{\DeltaX}{\Gamma}~\wedge~ \wf{\DeltaX}{S}~\wedge \\
			\quad \forall \rho, \gamma_1,\gamma_2.
			~\rho \in \gsetd{\DeltaX}. 
			~ (\gamma_1,\gamma_2) \in \gsetg{\Gamma} \implies (\rhofst(\gamma_1(e)),\rhosnd(\gamma_2(e))) \in \gsetc{S}
		\end{array}
\end{displaymath}
\end{definition}

\subsection{Security type soundness}

Instead of directly proving that the type system of Figure~\ref{fig:elsec-static-semantics} implies existential relaxed noninterference for all well-typed terms, we prove it through the standard definition of logically-related open terms~\cite{cruzAl:ecoop2017}:

\begin{definition}[Logically-related open terms]
\label{def:elsec-expression-equivalence}
\begin{displaymath}
	\begin{array}{l}
		\Delta ;\Gamma |- e_1 \approx e_2 : S   <=> \Delta;\Gamma |- e_i: S~\wedge~\wf{\Delta}{\Gamma}~\wedge~ \wf{\Delta}{S}~\wedge \\
		\quad \forall \rho, \gamma_1,\gamma_2.
		~\rho \in \gsetd{\Delta}. (\gamma_1, \gamma_2) \in \gsetg{\Gamma} \implies (\rhofst(\gamma_1(e_1)), \rhosnd(\gamma_2(e_2)) \in \gsetc{S}
	\end{array}
\end{displaymath}
\end{definition}

The next lemma captures that if an expression is logically related to itself, then it satisfies ERNI.

\begin{restatable}[Self logical relation implies PRNI]{lemma}{fpImpliesERNI}
\label{def:elsec-logapprox-implies-erni}
\mbox{}
\\
$\Delta;\Gamma |- e \approx e : S \implies \erni{\Delta}{\Gamma}{e}{S}$
\end{restatable}

The proof of security type soundness relies on the Fundamental Property of the logical relation: a well-typed \elsec term
is related to itself.

\begin{restatable}[Fundamental property]{theorem}{fperni}
\label{lm:elsec-fp}
$\Delta;\Gamma \vdash e : S \implies \Delta;\Gamma |- e \approx e : S$
\end{restatable}

Security type soundness follows from Lemma~\ref{def:elsec-logapprox-implies-erni} and Theorem~\ref{lm:elsec-fp}.

\begin{restatable}[Security type soundness]{theorem}{secTypeSoundness}
\label{lm:elsec-security-type-soundness}
$\Delta;\Gamma \vdash e : S \implies \erni{\Delta}{\Gamma}{e}{S}$
\end{restatable}
\begin{proof}
By induction on the typing derivation of $e$. 
Following \citet{ahmed:esop2006}, we define a compatibility lemma for each typing rule; then, each case of the induction directly follows from the corresponding compatibility lemma.
\end{proof}

\section{Illustration}
\label{sec:elsec-erni-formal-illustration}

With all the formal definitions at hand, we end by revisiting the informal example of Section~\ref{sec:elsec-erni-overview}.
$\SalaryPolicy$ can be encoded in \elsec as follow : 
$\eType{X}{\pairt{\twoTypes{\Int}{X}}{(\twoTypes{\Int}{X}-> \stypebot{\Bool})}}$. 
Let us show that the following two packages are equivalent at type $\stypebot{\SalaryPolicy}$ (\textsf{gte} is a curried comparison function, and we omit the \textsf{as}):
$$p_1~\defas~\packSimple{\Int}{\paire{\mathsf{100001}}{\mathsf{gte(100000)}}} \qquad
		p_2~\defas~\packSimple{\Int}{\paire{\mathsf{100002}}{\mathsf{gte(100000)}}} \\
	$$	
The definition of $\esetvx{\SalaryPolicy}{\emptyset}$
requires picking a relation $\XRel \in \erel{\Int}{\Int}$.
Pick $\XRel = \{(\mathsf{100001},\mathsf{100002})\}$. Then apply the rest of the definitions to verify that the package 
implementations are equivalent at 
$\esetvx{\pairt{\twoTypes{\Int}{X}}{(\twoTypes{\Int}{X} -> \stypebot{\Bool})}}{\xsrho{\emptyset}{X}{(\Int,\Int,\XRel)}}$. 
Use the $\rhosem{X}$ part of the definition of \\*
$\esetvx{\twoTypes{\Int}{X}}{\xsrho{\emptyset}{X}{(\Int,\Int,\XRel)}}$  to show that the first components 
$\mathsf{100001}$ and $\mathsf{100002}$  are
equivalent at $\twoTypes{\Int}{X}$, \ie $(\mathsf{100001},\mathsf{100002}) \in \rhosem{X}$.

 First we illustrate the formal reasoning that we obtain from Theorem~\ref{lm:elsec-security-type-soundness}. Let us pose $\Delta = X:\Int$ 
 and $\Gamma = x: \stypebot{(\pairt{\twoTypes{\Int}{X}}{(\twoTypes{\Int}{X}-> \stypebot{\Bool})})}$. The program 
 $e=  \appe{(\pSnd{x})}{(\pFst{x})}$ has type $\stypebot{\Bool}$, therefore, by Theorem~\ref{lm:elsec-security-type-soundness}, $\erni{\Delta}{\Gamma}{e}{\stypebot{\Bool}}$ holds---$e$ is secure at $\stypebot{\Bool}$. We can verify this formally.
 By Definition~\ref{def:elsec-erni} we have to assume an arbitrary type substitution $\rho \in \esetd{X:\Int}$ and two values substitutions 
 $(\gamma_1,\gamma_2) \in \esetg{x:\stypebot{(\pairt{\twoTypes{\Int}{X}}{(\twoTypes{\Int}{X}-> \stypebot{\Bool})})}}$ and to show 
 that $(\rho_1(\gamma_1(\appe{(\pSnd{x})}{(\pFst{x})})),\rho_2(\gamma_2(\appe{(\pSnd{x})}{(\pFst{x})}))) \in \esetc{\stypebot{\Bool}}$. From $(\gamma_1,\gamma_2) \in \esetg{x:\stypebot{(\pairt{\twoTypes{\Int}{X}}{(\twoTypes{\Int}{X}-> \stypebot{\Bool})})}}$ we know that $\gamma_1 = x \mapsto v_1$ and 
 $\gamma_1 = x \mapsto v_2$, such as $(v_1,v_2) \in \esetv{\stypebot{(\pairt{\twoTypes{\Int}{X}}{(\twoTypes{\Int}{X}-> \stypebot{\Bool})})}}$. 
 Then $\appe{(\pSnd{v_1})}{(\pFst{v_1})} \reduce^{*} \appe{v'_{11}}{v'_{12}}$ and 
 $\appe{(\pSnd{v_2})}{(\pFst{v_2})} \reduce^{*} \appe{v'_{21}}{v'_{22}}$, such that 
 $(v'_{11},v'_{21}) \in \esetvx{\twoTypes{\Int}{X}-> \stypebot{\Bool}}{\rho}$ and $(v'_{12},v'_{22}) \in 
 \esetvx{\twoTypes{\Int}{X}}{\rho}$. After the previous reductions we have to verify that $(\appe{v'_{11}}{v'_{12}},\appe{v'_{21}}{v'_{22}}) \in \esetc{\stypebot{\Bool}}$.
 At this point, instantiate $(v'_{11},v'_{21}) \in 
 \esetvx{\twoTypes{\Int}{X}-> \stypebot{\Bool}}{\rho}$ with $(v'_{12},v'_{22}) \in 
 \esetvx{\twoTypes{\Int}{X}}{\rho}$ to obtain $(\appe{v'_{11}}{v'_{12}},\appe{v'_{21}}{v'_{22}}) \in \esetc{\stypebot{\Bool}}$.

 Second, we formally show why $\erni{\Delta}{\Gamma}{(\pFst{x}) \% 2}{\stypebot{\Int}}$ does not hold. 
 Instantiate Definition~\ref{def:elsec-erni} with $\rho = X \mapsto (\Int,\Int,R)$ and $\gamma_1 = x \mapsto \paire{\mathsf{100001}}{\mathsf{gte(100000)}}$ and 
 $\gamma_2 = x \mapsto \paire{\mathsf{100002}}{\mathsf{gte(100000)}}$. 
 Note that $\rho =  X \mapsto (\Int,\Int,R) \in \esetd{X:\Int}$ and $(\gamma_1,\gamma_2) \in \esetg{\Gamma}$.
 To show $(\rho_1(\gamma_1((\pFst{x}) \% 2)),\rho_2(\gamma_2((\pFst{x}) \% 2))) \in \esetc{\stypebot{\Int}}$ requires showing 
 $(\mathsf{100001 \% 2},\mathsf{100002 \% 2}) \in \esetc{\stypebot{\Int}}$, which means showing $(1,0) \in \esetc{\stypebot{\Int}}$---which is false. Similarly, we can verify that  $\erni{\Delta}{\Gamma}{\mathsf{x.salary}}{\stypebot{\Int}}$ does not hold, which means that a declassifiable secret cannot be directly observed by a public observer.

\section{Related Work}
\label{sec:elsec-related-work}
%

We have already extensively discussed the relation to the original formulation 
of the labels-as-types approach in an object-oriented setting~\cite{cruzAl:ecoop2017}, itself inspired by the work on declassification policies (labels-as-functions) of \citet{liZdancewic:popl2005}.
Formulating type-based declassification with existential types shows how to exploit another type abstraction mechanism that is found in non-object-oriented languages, with abstract data types and modules. Also, existential types support extrinsic 
declassification policies, which are not expressible in the receiver-centric approach of objects.
For instance, the $\AccountStore$ example of Section~\ref{sec:elsec-overview} is not supported by design in the object-oriented approach.

The extrinsic declassification policies supported by our approach are closely related to {\em trusted declassification}~\cite{hicksAl:plas2006}, 
where declassification is globally defined, associating principals that own  secrets with trusted (external) methods that can declassify these secrets. 
In our approach, the relation between secrets and declassifiers is not globally defined, but is local to an existential type and its usage. 
In both approaches the implementations of declassifiers have a privileged view of the secrets.

\citet{bowmanAhmed:icfp2015} present a translation of noninterference into parametricity with a compiler from the Dependency Core Calculus (DCC)~\cite{abadiAl:popl99} to System~F$\omega$. In a recent (as yet unpublished) article, \citet{ngoAl:arXiv2019} extend this work to support translating declassification policies, inspired by prior work on type-based relaxed noninterference~\cite{cruzAl:ecoop2017}. They first provide a translation into abstract types 
of the polymorphic lambda calculus~\cite{reynolds:83}, and then into signatures of a module calculus~\cite{crary:popl2017}. 
While that work and ours encode declassification policies via existential types (module signatures), we focus on providing a surface language for information flow control with type-based declassification. In particular, their translated programs do not support computing with secrets, which is enabled in both this work and the original work of \citet{cruzAl:ecoop2017} thanks to faceted types. Additionally, they only model first-order secrets (integers), while our modular reasoning principle seamlessly accommodates higher-order secrets.

In another very recent piece of work, \citet{cruzTanter:secdev2019} extend the object-oriented approach to type-based relaxed noninterference with 
parametric polymorphism, thereby supporting polymorphic declassification policies. Polymorphic declassification for object types 
is achieved with type variables at the method signature level, which supports the specification of polymorphic policies of the form $\stype{T}{X}$.
Existential types are closely related to universal types.  
In particular, the client of a package that exports a type variable $X$ must 
be polymorphic with respect to $X$; hence our work supports a form of declassification polymorphism in the client code. It would be interesting to extend \elsec with universal types in order to study the interaction of both abstraction mechanisms in a standard functional setting. 
Finally, because of the receiver-centric perspective of objects, they have to resort to ad-hoc polymorphism to properly account for primitive types. Here, 
primitive types do not require any special treatment for declassification polymorphism, because of our extrinsic approach to declassification.

The idea of using the abstraction mechanism of modules to express a form of declassification can also be found in the work of~\citet{nanevskiAl:sp2011} on 
Relational Hoare Type Theory (RHTT). RHTT is formulated with a monadic security type constructor $\mathsf{STsec}~A (p,q)$, where $p$ is a pre-condition on the heap, and $q$ is a post-condition relating output values, input heaps and output heaps. Thanks to the expressive power of the underlying dependent type theory, preconditions and postconditions can characterize very precise declassification policies. The price to pay is that proofs of noninterference have to be provided explicitly as proof terms (or discharged via tactics or other means when possible), while our less expressive approach is a simple, non-dependent type system. Finding the right balance between the expressiveness and the complexity of the typing discipline to express security policies is an active subject of research.

\section{Conclusion}
\label{sec:elsec-conclusion}
We present a novel approach to type-based relaxed noninterference, based on existential types as the underlying type abstraction mechanism. In contrast to the object-oriented, subtyping-based approach, the existential approach naturally supports external declassification policies. This work shows that the general approach of faceted security types for expressive declassification can be applied in non-object-oriented languages that support abstract data types or modules. As such, it represents a step towards providing a practical realization of information-flow security typing that
accounts for controlled and expressive declassification with a modular reasoning principle about security.

An immediate venue for future work that would be crucial in practice is to develop type inference for declassification types, which should reduce to standard type inference~\cite{damasMilner:popl1982}. Finally, a particularly interesting perspective is to study the combination of the existential approach with the object-oriented approach, thereby bridging the gap towards a practical implementation in a full-fledged programming language like Scala that features all these type abstraction mechanisms.

{\small 
\renewcommand{\bibsection}{\section{References}}
\bibliographystyle{splncsnat}


\bibliography{extracted}


\end{document}